\documentclass{llncs}
\usepackage{version}
\usepackage{hl}

\title{A Hoare-Like Logic of\\ Asserted Single-Pass Instruction Sequences}
\author{J.A. Bergstra \and C.A. Middelburg}
\institute{Informatics Institute, Faculty of Science, University of
           Amsterdam \\[1ex]
           Science Park~904, 1098~XH Amsterdam, the Netherlands \\[1ex]
           \email{J.A.Bergstra@uva.nl,C.A.Middelburg@uva.nl}}

\begin{document}
\maketitle

\begin{abstract}
We present a formal system for proving the partial correctness of a 
single-pass instruction sequence as considered in program algebra by 
decomposition into proofs of the partial correctness of segments of the 
single-pass instruction sequence concerned.
The system is similar to Hoare logics, but takes into account that, by 
the presence of jump instructions, segments of single-pass instruction 
sequences may have multiple entry points and multiple exit points.
It is intended to support a sound general understanding of the issues 
with Hoare-like logics for low-level programming languages.
\begin{keywords} 
Hoare logic, asserted single-pass instruction sequence, 
soundness, completeness in the sense of Cook.
\end{keywords}%
\begin{classcode}
D.2.4, F.3.1
\end{classcode}
\end{abstract}

\section{Introduction}
\label{sect-intro}

In~\cite{Hoa69a}, Hoare introduced a kind of formal system for proving 
the partial correctness of a program by decomposition into proofs of 
the partial correctness of segments of the program concerned.
Formal systems of this kind are now known as Hoare logics.
The programs considered in~\cite{Hoa69a} are programs in a simple 
high-level programming language without goto statements.
Hoare logics for this simple high-level programming language and 
extensions of it without goto statements have been extensively studied 
since (see e.g.~\cite{BT82b,Cla79a,Coo78a} for individual studies 
and~\cite{Apt81a} for a survey). 
Hoare logics and Hoare-like logics for simple high-level programming 
languages with goto statements have been studied since as well (see 
e.g.~\cite{Bru81a,CH72a,Wan76a}).

Work on Hoare-like logics for low-level programming languages started 
only recently. 
All the work that we know of takes ad hoc restrictions and features of 
machine- or assembly-level programs into account (see e.g.~\cite{MG07a}) 
or abstracts in an ad hoc way from instruction sequences as found in 
low-level programs (see e.g.~\cite{SU07a}).
We consider it important for a sound understanding of the issues in this 
area to give consideration to generality and faithfulness of abstraction
instead. 
This is what motivated us to do the work presented in this paper.

We present a Hoare-like logic for single-pass instruction sequences as 
considered in~\cite{BL02a}. 
The instruction sequences in question are finite or eventually periodic 
infinite sequences of instructions of which each instruction is executed 
at most once and can be dropped after it has been executed or jumped 
over.
We will come back to the choice for those instruction sequences. 
The presented Hoare-like logic has to take into account that, by the 
presence of jump instructions, segments of instruction sequences may 
have multiple entry points and multiple exit points.
Because of this, it is closer to the inductive assertion method for 
program flowcharts introduced by Floyd in~\cite{Flo67a} than most other
Hoare and Hoare-like logics.

The asserted programs of the form $\assprog{P}{S}{Q}$ of Hoare logics 
are replaced in the presented Hoare-like logic by asserted instruction 
sequences of the form $\assinseq{b}{P}{S}{e}{Q}$, where $b$ and $e$ are 
a positive natural number and a natural number, respectively.
$P$ and $Q$ are regular pre- and post-conditions.
That is, they concern the input-output behaviour of the instruction 
sequence segment $S$. 
Loosely speaking, $b$ represents the additional pre-condition that 
execution enters the instruction sequence segment $S$ at its $b$th 
instruction and, if $e$ is positive, $e$ represents the additional 
post-condition that either execution exits the instruction sequence 
segment $S$ by going to the $e$th instruction following it or becomes 
inactive in $S$.
In the case that $e$ equals zero, $e$ represents the addi\-tional 
post-condition that execution either terminates or becomes inactive in 
$S$ (instructions sequences with explicit termination instructions are
considered).

The form of the asserted instruction sequences is inspired 
by~\cite{Wan76a}.
However, under the interpretation of~\cite{Wan76a}, $e$ would represent 
the additional post-condition that execution reaches the $e$th 
instruction following the first instruction of the instruction sequence 
segment concerned.
Because this may be an instruction before the first instruction 
following the segment, this interpretation allows of asserted 
instruction sequences that concern the internals of the segment.
For this reason, we consider this interpretation not conducive to 
compositional proofs.

In other related work, e.g.\ in~\cite{SU07a}, the additional pre- and 
post-condition represented by $b$ and $e$ must be explicitly formulated
and conjoined with the regular pre- and post-condition, respectively.
This alternative reduces the conciseness of pre- and post-conditions
considerably.
Moreover, an effect ensuing from this alternative is that assertions can 
be formulated in which aspects of input-output behaviour and flow of 
execution are combined in ways that are unnecessary for proving partial 
correctness.
For these reasons, we decided not to opt for this~alternative.

There is a tendency in work on Hoare-like logics to use a separation 
logic instead of classical first-order logic for pre- and 
post-conditions to deal with programs that alter data structures
(see e.g.~\cite{Rey02a}).
This tendency is also found in work on Hoare-like logics for low-level 
programming languages (see e.g.~\cite{JBK13a}).
Because our intention is to present a Hoare-like logic that supports a 
sound general understanding of the issues with Hoare-like logics for 
low-level programming languages, we believe that we should stick to 
classical first-order logic until it proves to be inadequate.
This is the reason why classical first-order logic is used for pre- and 
post-conditions in this paper.

As mentioned before, the presented Hoare-like logic concerns single-pass 
instruction sequences as considered in~\cite{BL02a}. 
It is often said that a program is an instruction sequence and, if this 
characterization has any value, it must be the case that it is somehow 
easier to understand the concept of an instruction sequence than to 
understand the concept of a program. 
The first objective of the work on instruction sequences that started 
with~\cite{BL02a}, and of which an enumeration is available 
at~\cite{SiteIS}, is to understand the concept of a program.
The basis of all this work is an algebraic theory of single-pass 
instruction sequences, called program algebra, and an algebraic theory 
of mathematical objects that represent in a direct way the behaviours 
produced by instruction sequences under execution, called basic thread 
algebra.%
\footnote
{In~\cite{BL02a}, basic thread algebra is introduced under the name
 basic polarized process algebra.}
The body of theory developed through this work is such that its use as a
conceptual preparation for programming is practically feasible.

The notion of an instruction sequence appears in the work in question as 
a mathematical abstraction for which the rationale is based on the 
objective mentioned above. 
In this capacity, instruction sequences constitute a primary field of 
investigation in programming comparable to propositions in logic and 
rational numbers in arithmetic. 
The structure of the mathematical abstraction at issue has been 
determined in advance with the hope of applying it in diverse 
circumstances where in each case the fit may be less than perfect.
Until now, this work has, among other things, yielded an approach to 
computational complexity where program size is used as complexity 
measure, a contribution to the conceptual analysis of the notion of an 
algorithm, and new insights into such diverse issues as the halting 
problem, garbage collection, program parallelization for the purpose of 
explicit multi-threading and virus detection.

Judging by our experience gained in the work referred to above, we think 
that generality and faithfulness of abstraction are well taken into 
consideration in a Hoare-like logic for single-pass instruction 
sequences as considered in~\cite{BL02a}.
This explains the choice for those instruction sequences.
As in the work referred to above, the work presented in this paper is 
carried out in the setting of program algebra and basic thread algebra.

This paper is organized as follows.
First, we give a survey of program algebra and basic thread algebra 
(Section~\ref{sect-PGA-and-BTA}) and a survey of the extension of basic 
thread algebra that is used in this paper (Section~\ref{sect-TSI}).
Next, we present a Hoare-like logic of asserted single-pass instruction 
sequences (Section~\ref{sect-HL-PC}), give an example of its use  
(Section~\ref{sect-example}), and show that it is sound and complete in 
the sense of Cook (Section~\ref{sect-sound-complete}). 
Finally, we make some concluding remarks (Section~\ref{sect-concl}).

Some familiarity with algebraic specification is assumed in this paper.
The relevant notions are explained in handbook chapters and books on 
algebraic specification, e.g.~\cite{EM85a,ST99a,ST12a,Wir90a}.

The preliminaries to the work presented in this paper 
(Sections~\ref{sect-PGA-and-BTA} and~\ref{sect-TSI}) are almost the 
same as the preliminaries to the work presented in~\cite{BM15a} and 
earlier papers.
For this reason, there is some text overlap with those papers.
Apart from the preliminaries, the material in this paper is new.
A comprehensive introduction to what is surveyed in the preliminary 
sections can among other things be found in~\cite{BM12b}.

\section{Program Algebra and Basic Thread Algebra}
\label{sect-PGA-and-BTA}

In this section, we give a survey of \PGA\ (ProGram Algebra) and \BTA\ 
(Basic Thread Algebra) and make precise in the setting of \BTA\ which 
behaviours are produced by the instruction sequences considered in \PGA\
under execution.
The greater part of this section originates from~\cite{BM13a}.

In \PGA, it is assumed that there is a fixed but arbitrary set $\BInstr$
of \emph{basic instructions}.
The intuition is that the execution of a basic instruction may modify a 
state and produces a reply at its completion.
The possible replies are $\False$ and $\True$.
The actual reply is generally state-dependent.
The set $\BInstr$ is the basis for the set of instructions that may 
occur in the instruction sequences considered in \PGA.
The elements of the latter set are called \emph{primitive instructions}.
There are five kinds of primitive instructions:
\begin{itemize}
\item
for each $a \in \BInstr$, a \emph{plain basic instruction} $a$;
\item
for each $a \in \BInstr$, a \emph{positive test instruction} $\ptst{a}$;
\item
for each $a \in \BInstr$, a \emph{negative test instruction} $\ntst{a}$;
\item
for each $l \in \Nat$, a \emph{forward jump instruction} $\fjmp{l}$;
\item
a \emph{termination instruction} $\halt$.
\end{itemize}
We write $\PInstr$ for the set of all primitive instructions.

On execution of an instruction sequence, these primitive instructions
have the following effects:
\begin{itemize}
\item
the effect of a positive test instruction $\ptst{a}$ is that basic
instruction $a$ is executed and execution proceeds with the next
primitive instruction if $\True$ is produced and otherwise the next
primitive instruction is skipped and execution proceeds with the
primitive instruction following the skipped one --- if there is no
primitive instruction to proceed with, execution becomes inactive;
\item
the effect of a negative test instruction $\ntst{a}$ is the same as
the effect of $\ptst{a}$, but with the role of the value produced
reversed;
\item
the effect of a plain basic instruction $a$ is the same as the effect
of $\ptst{a}$, but execution always proceeds as if $\True$ is produced;
\item
the effect of a forward jump instruction $\fjmp{l}$ is that execution
proceeds with the $l$th next primitive instruction --- if $l$ equals $0$ 
or there is no primitive instruction to proceed with, execution becomes 
inactive;
\item
the effect of the termination instruction $\halt$ is that execution 
terminates.
\end{itemize}
Execution becomes inactive if no more basic instructions are executed, 
but execution does not terminate.

\PGA\ has one sort: the sort $\InSeq$ of \emph{instruction sequences}. 
We make this sort explicit to anticipate the need for many-sortedness
later on.
To build terms of sort $\InSeq$, \PGA\ has the following constants and 
operators:
\begin{itemize}
\item
for each $u \in \PInstr$, 
the \emph{instruction} constant $\const{u}{\InSeq}$\,;
\item
the binary \emph{concatenation} operator 
$\funct{\ph \conc \ph}{\InSeq \x \InSeq}{\InSeq}$\,;
\item
the unary \emph{repetition} operator 
$\funct{\ph\rep}{\InSeq}{\InSeq}$\,.
\end{itemize} 
Terms of sort $\InSeq$ are built as usual in the one-sorted case.%
\footnote{Notice that all \PGA\ term are of sort $\InSeq$.}
We assume that there are infinitely many variables of sort $\InSeq$, 
including $X,Y,Z$.
We use infix notation for concatenation and postfix notation for
repetition.
Hence, taking these notational conventions into account, the syntax of 
closed terms of sort $\InSeq$ can be defined in Backus-Naur style as 
follows:
\begin{ldispl}
\nm{CT} \Sis 
a \Sor \ptst{a} \Sor \ntst{a} \Sor \fjmp{l} \Sor \halt \Sor 
\nm{CT} \conc \nm{CT} \Sor \nm{CT} \rep\;, 
\end{ldispl}%
where $a \in \BInstr$ and $l \in \Nat$.

A closed \PGA\ term is considered to denote a non-empty, finite or
eventually periodic infinite sequence of primitive instructions.%
\footnote
{An eventually periodic infinite sequence is an infinite sequence with
 only finitely many distinct suffixes.}
The instruction sequence denoted by a closed term of the form
$t \conc t'$ is the instruction sequence denoted by $t$
concatenated with the instruction sequence denoted by $t'$.
The instruction sequence denoted by a closed term of the form $t\rep$
is the instruction sequence denoted by $t$ concatenated infinitely
many times with itself.
A simple example of a closed PGA term is
\begin{ldispl}
(\ptst{a} \conc \fjmp{2} \conc \fjmp{3} \conc b \conc \halt)\rep\;.
\end{ldispl}%
On execution of the instruction sequence denoted by this term, first the 
basic instruction $a$ is executed repeatedly until its execution 
produces the reply $\True$, next the basic instruction $b$ is executed, 
and after that execution terminates.

Closed \PGA\ terms are considered equal if they represent the same
instruction sequence.
The axioms for instruction sequence equivalence are given in
Table~\ref{axioms-PGA}.%
\begin{table}[!t]
\caption{Axioms of \PGA}
\label{axioms-PGA}
\begin{eqntbl}
\begin{axcol}
(X \conc Y) \conc Z = X \conc (Y \conc Z)              & \axiom{PGA1} \\
(X^n)\rep = X\rep                                      & \axiom{PGA2} \\
X\rep \conc Y = X\rep                                  & \axiom{PGA3} \\
(X \conc Y)\rep = X \conc (Y \conc X)\rep              & \axiom{PGA4}
\end{axcol}
\end{eqntbl}
\end{table}
In this table, $n$ stands for an arbitrary positive natural number.
For each natural number $n$, the term $t^n$, where $t$ is a \PGA\ term, 
is defined by induction on $n$ as follows: $t^0 = \fjmp{0}$, $t^1 = t$, 
and $t^{n+2} = t \conc t^{n+1}$.
Some simple examples of equations derivable from the axioms of \PGA\ are
\begin{ldispl}
(a \conc b)\rep \conc \halt = a \conc (b \conc a)\rep\;,
\\
\ptst{a} \conc (b \conc (\ntst{c} \conc \fjmp{2} \conc \halt)\rep)\rep
=
\ptst{a} \conc b \conc (\ntst{c} \conc \fjmp{2} \conc \halt)\rep\;.
\end{ldispl}%

A typical model of \PGA\ is the model in which:
\begin{itemize}
\item
the domain is the set of all finite and eventually periodic infinite
sequences over the set $\PInstr$ of primitive instructions;
\item
the operation associated with ${} \conc {}$ is concatenation;
\item
the operation associated with ${}\rep$ is the operation ${}\srep$
defined as follows:

\begin{itemize}
\item
if $U$ is a finite sequence over $\PInstr$, then $U\srep$ is the unique 
infinite sequence $U'$ such that $U$ concatenated $n$ times with itself 
is a proper prefix of $U'$ for each $n \in \Nat$;
\item
if $U$ is an infinite sequence over $\PInstr$, then $U\srep$ is $U$.
\end{itemize}
\end{itemize}
We confine ourselves to this model of \PGA, which is an initial model of 
\PGA, for the interpretation of \PGA\ terms.
In the sequel, we use the term \emph{PGA instruction sequence} for the 
elements of the domain of this model and write $\len(t)$, where $t$ is a 
closed \PGA\ term denoting a finite \PGA\ instruction sequence, for the 
length of the \PGA\ instruction sequence denoted by $t$.
We stipulate that $\len(t) = \omega$ if $t$ is a closed \PGA\ term 
denoting an infinite instruction sequence, where $n < \omega$ for all 
$n \in \Nat$.

Below, we will use \BTA\ to make precise which behaviours are produced 
by \PGA\ instruction sequences under execution.

In \BTA, it is assumed that a fixed but arbitrary set $\BAct$ of
\emph{basic actions} has been given.
The objects considered in \BTA\ are called threads.
A thread represents a behaviour which consists of performing basic 
actions in a sequential fashion.
Upon each basic action performed, a reply from an execution environment
determines how the thread proceeds.
The possible replies are the values $\False$ and $\True$.

\BTA\ has one sort: the sort $\Thr$ of \emph{threads}. 
We make this sort explicit to anticipate the need for many-sortedness
later on.
To build terms
of sort $\Thr$, \BTA\ has the following constants and operators:
\begin{itemize}
\item
the \emph{inaction} constant $\const{\DeadEnd}{\Thr}$;
\item
the \emph{termination} constant $\const{\Stop}{\Thr}$;
\item
for each $a \in \BAct$, the binary \emph{postconditional composition} 
operator $\funct{\pcc{\ph}{a}{\ph}}{\linebreak[2]\Thr \x \Thr}{\Thr}$.
\end{itemize}
Terms of sort $\Thr$ are built as usual in the one-sorted case. 
We assume that there are infinitely many variables of sort $\Thr$, 
including $x,y$.
We use infix notation for postconditional composition. 
We introduce \emph{basic action prefixing} as an abbreviation: 
$a \bapf t$, where $t$ is a \BTA\ term, abbreviates 
$\pcc{t}{a}{t}$.
We identify expressions of the form $a \bapf t$ with the \BTA\
term they stand for.

The thread denoted by a closed term of the form $\pcc{t}{a}{t'}$
will first perform $a$, and then proceed as the thread denoted by
$t$ if the reply from the execution environment is $\True$ and proceed
as the thread denoted by $t'$ if the reply from the execution
environment is $\False$. 
The thread denoted by $\Stop$ will do no more than terminate and the 
thread denoted by $\DeadEnd$ will become inactive.
A simple example of a closed \BTA\ term is
\begin{ldispl}
\pcc{(b \bapf \Stop)}{a}{\DeadEnd}\;.
\end{ldispl}%
This term denotes the thread that first performs basic action
$a$, if the reply from the execution environment on performing $a$ is
$\True$, next performs the basic action $b$ and after that terminates,
and if the reply from the execution environment on performing $a$ is
$\False$, next becomes inactive.

Closed \BTA\ terms are considered equal if they are syntactically the
same.
Therefore, \BTA\ has no axioms.

Each closed \BTA\ term denotes a finite thread, i.e.\ a thread with a
finite upper bound to the number of basic actions that it can perform.
Infinite threads, i.e.\ threads without a finite upper bound to the
number of basic actions that it can perform, can be defined by means of 
a set of recursion equations (see e.g.~\cite{BM09k}).
We are only interested in models of \BTA\ in which sets of recursion 
equations have unique solutions, such as the projective limit model 
of \BTA\ presented in~\cite{BM12b}.
We confine ourselves to this model of \BTA, which has an initial model 
of \BTA\ as a submodel, for the interpretation of \BTA\ terms. 
In the sequel, we use the term \emph{BTA thread} or simply \emph{thread} 
for the elements of the domain of this model.

Regular threads, i.e.\ finite or infinite threads that can only be in a 
finite number of states, can be defined by means of a finite set of 
recursion equations.
The behaviours produced by \PGA\ instruction sequences under execution 
are exactly the behaviours represented by regular threads, with the basic 
instructions taken for basic actions.
The behaviours produced by finite \PGA\ instruction sequences under 
execution are the behaviours represented by finite threads.

We combine \PGA\ with \BTA\ and extend the combination with
the \emph{thread extraction} operator $\funct{\extr{\ph}}{\InSeq}{\Thr}$, 
the axioms given in Table~\ref{axioms-thread-extr},%
\begin{table}[!tb]
\caption{Axioms for the thread extraction operator}
\label{axioms-thread-extr}
\begin{eqntbl}
\begin{eqncol}
\extr{a} = a \bapf \DeadEnd \\
\extr{a \conc X} = a \bapf \extr{X} \\
\extr{\ptst{a}} = a \bapf \DeadEnd \\
\extr{\ptst{a} \conc X} =
\pcc{\extr{X}}{a}{\extr{\fjmp{2} \conc X}} \\
\extr{\ntst{a}} = a \bapf \DeadEnd \\
\extr{\ntst{a} \conc X} =
\pcc{\extr{\fjmp{2} \conc X}}{a}{\extr{X}}
\end{eqncol}
\qquad
\begin{eqncol}
\extr{\fjmp{l}} = \DeadEnd \\
\extr{\fjmp{0} \conc X} = \DeadEnd \\
\extr{\fjmp{1} \conc X} = \extr{X} \\
\extr{\fjmp{l+2} \conc u} = \DeadEnd \\
\extr{\fjmp{l+2} \conc u \conc X} = \extr{\fjmp{l+1} \conc X} \\
\extr{\halt} = \Stop \\
\extr{\halt \conc X} = \Stop
\end{eqncol}
\end{eqntbl}
\end{table}
and the rule that $\extr{X} = \DeadEnd$ if $X$ has an infinite chain of 
forward jumps beginning at its first primitive instruction.%
\footnote
{This rule, which can be formalized using an auxiliary structural 
congruence predicate (see e.g.~\cite{BM07g}), is unnecessary when 
considering only finite \PGA\ instruction sequences.
}
In Table~\ref{axioms-thread-extr}, $a$ stands for an arbitrary basic 
instruction from $\BInstr$, $u$ stands for an arbitrary primitive 
instruction from $\PInstr$, and $l$ stands for an arbitrary natural 
number from $\Nat$.
For each closed \PGA\ term $t$, $\extr{t}$ denotes the behaviour  
produced by the instruction sequence denoted by $t$ under execution.

A simple example of thread extraction is
\begin{ldispl}
\extr{\ptst{a} \conc \fjmp{2} \conc \fjmp{3} \conc b \conc \halt} =
\pcc{(b \bapf \Stop)}{a}{\DeadEnd}\;,
\end{ldispl}%
In the case of infinite instruction sequences, thread extraction yields 
threads definable by means of a set of recursion equations.
For example, 
\begin{ldispl}
\extr{(\ptst{a} \conc \fjmp{2} \conc \fjmp{3} \conc b \conc \halt)\rep}
\end{ldispl}%
is the solution of the set of recursion equations that consists of the 
single equation
\begin{ldispl}
x = \pcc{(b \bapf \Stop)}{a}{x}\;.
\end{ldispl}%

\section{Interaction of Threads with Services}
\label{sect-TSI}

Services are objects that represent the behaviours exhibited by 
components of execution environments of instruction sequences at a high 
level of abstraction.
A service is able to process certain methods.
The processing of a method may involve a change of the service.
At completion of the processing of a method, the service produces a
reply value.
For example, a service may be able to process methods for pushing a
natural number on a stack ($\push{n}$), popping the top element from the 
stack ($\pop$), and testing whether the top element of the stack equals 
a natural number ($\topeq{n}$).
Processing of a pushing method or a popping method changes the service 
and produces the reply value $\True$ if no stack underflow occurs and 
$\False$ otherwise.
Processing of a testing method does not change the service and produces 
the reply value $\True$ if the test succeeds and $\False$ otherwise.

Execution environments are considered to provide a family of 
uniquely-named services.
A thread may interact with the named services from the service family 
provided by an execution environment.
That is, a thread may perform a basic action for the purpose of 
requesting a named service to process a method and to return a reply 
value at completion of the processing of the method.
In this section, we extend \BTA\ with services, service families, a 
composition operator for service families, and an operator that is 
concerned with this kind of interaction.
This section originates from~\cite{BM09k}.

In \SFA, the algebraic theory of service families introduced below, it 
is assumed that a fixed but arbitrary set $\Meth$ of \emph{methods} has 
been given.
Moreover, the following is assumed with respect to services:
\begin{itemize}
\item
a signature $\Sig{\ServAlg}$ has been given that includes the following
sorts:
\begin{itemize}
\item
the sort $\Serv$ of
\emph{services};
\item
the sort $\Repl$ of \emph{replies};
\end{itemize}
and the following constants and operators:
\begin{itemize}
\item
the
\emph{empty service} constant $\const{\emptyserv}{\Serv}$;
\item
the \emph{reply} constants $\const{\False,\True,\Div}{\Repl}$;
\item
for each $m \in \Meth$, the
\emph{derived service} operator $\funct{\derive{m}}{\Serv}{\Serv}$;
\item
for each $m \in \Meth$, the
\emph{service reply} operator $\funct{\sreply{m}}{\Serv}{\Repl}$;
\end{itemize}
\item
a minimal $\Sig{\ServAlg}$-algebra $\ServAlg$ has been given in which
$\False$, $\True$, and $\Div$ are mutually different, and
\begin{itemize}
\item
$\LAND{m \in \Meth}{}
  \derive{m}(z) = z \Land \sreply{m}(z) = \Div \Limpl
  z = \emptyserv$
holds;
\item
for each $m \in \Meth$,
$\derive{m}(z) = \emptyserv \Liff \sreply{m}(z) = \Div$ holds.
\end{itemize}
\end{itemize}

The intuition concerning $\derive{m}$ and $\sreply{m}$ is that on a
request to service $s$ to process method $m$:
\begin{itemize}
\item
if $\sreply{m}(s) \neq \Div$, $s$ processes $m$, produces the reply
$\sreply{m}(s)$, and then proceeds as $\derive{m}(s)$;
\item
if $\sreply{m}(s) = \Div$, $s$ is not able to process method $m$ and
proceeds as $\emptyserv$.
\end{itemize}
The empty service $\emptyserv$ itself is unable to process any method.

The actual services could, for example, be the natural number stack 
services sketched at the beginning of this section.
In that case, we take the set
$\set{\NNS_{\sigma} \where \sigma \in \seqof{\Nat}}$ of 
\emph{natural number stack services} as the set $\Services$ of services
and, for each $m \in \Meth$, we take the functions $\derive{m}$ and
$\sreply{m}$ such that ($n,n' \in \Nat$, $\sigma \in \seqof{\Nat}$):%
\footnote
{We write $\emptyseq$ for the empty sequence and $n \sigma$ for the
 sequence $\sigma$ with $n$ prepended to it.}
\begin{ldispl}
\begin{geqns}
\derive{\push{n}}(\NNS_{\sigma})     = \NNS_{n \sigma}\;,
\\
\derive{\pop}(\NNS_{n' \sigma})      = \NNS_{\sigma}\;,
\\
\derive{\pop}(\NNS_{\emptyseq})      = \NNS_{\emptyseq}\;,
\\
\derive{\topeq{n}}(\NNS_{n' \sigma}) = \NNS_{n' \sigma}\;,
\\ {} \\
\derive{\topeq{n}}(\NNS_{\emptyseq}) = \NNS_{\emptyseq}\;,
\\
\derive{m}(\NNS_{\sigma}) = \emptyserv \;\; \mif m \notin \Meth_\NNS\;,
\end{geqns}
\quad
\begin{gceqns}
\sreply{\push{n}}(\NNS_{\sigma})     = \True\;,
\\
\sreply{\pop}(\NNS_{n' \sigma})      = \True\;,
\\
\sreply{\pop}(\NNS_{\emptyseq})      = \False\;,
\\
\sreply{\topeq{n}}(\NNS_{n' \sigma}) = \True  & \; \mif n = n'\;,
\\
\sreply{\topeq{n}}(\NNS_{n' \sigma}) = \False & \; \mif n \neq n'\;,
\\
\sreply{\topeq{n}}(\NNS_{\emptyseq}) = \False\;,
\\
\sreply{m}(\NNS_{\sigma}) = \Div       & \; \mif m \notin \Meth_\NNS\;,
\end{gceqns}
\end{ldispl}%
where 
$\Meth_\NNS =
 \set{\push{n} \where n \in \Nat} \union \set{pop} \union
 \set{\topeq{n} \where n \in \Nat}$.

It is also assumed that a fixed but arbitrary set $\Foci$ of
\emph{foci} has been given.
Foci play the role of names of services in a service family. 

\SFA\ has the sorts, constants and operators from $\Sig{\ServAlg}$ and
in addition the sort $\ServFam$ of \emph{service families} and the 
following constant and operators:
\begin{itemize}
\item
the
\emph{empty service family} constant $\const{\emptysf}{\ServFam}$;
\item
for each $f \in \Foci$, the unary
\emph{singleton service family} operator
$\funct{\mathop{f{.}} \ph}{\Serv}{\ServFam}$;
\item
the binary
\emph{service family composition} operator
$\funct{\ph \sfcomp \ph}{\ServFam \x \ServFam}{\ServFam}$;
\item
for each $F \subseteq \Foci$, the unary
\emph{encapsulation} operator $\funct{\encap{F}}{\ServFam}{\ServFam}$.
\end{itemize}
We assume that there are infinitely many variables of sort $\Serv$,
including $z$, and infinitely many variables of sort $\ServFam$,
including $u,v,w$.
Terms are built as usual in the many-sorted case
(see e.g.~\cite{ST99a,Wir90a}).
We use prefix notation for the singleton service family operators and
infix notation for the service family composition operator.
We write $\Sfcomp{i = 1}{n} t_i$, where $t_1,\ldots,t_n$ are
terms of sort $\ServFam$, for the term
$t_1 \sfcomp \ldots \sfcomp t_n$.

The service family denoted by $\emptysf$ is the empty service family.
The service family denoted by a closed term of the form $f.t$ consists 
of one named service only, the service concerned is the service denoted 
by $t$, and the name of this service is $f$.
The service family denoted by a closed term of the form
$t \sfcomp t'$ consists of all named services that belong to either the
service family denoted by $t$ or the service family denoted by $t'$.
In the case where a named service from the service family denoted by
$t$ and a named service from the service family denoted by $t'$ have
the same name, they collapse to an empty service with the name
concerned.
The service family denoted by a closed term of the form $\encap{F}(t)$ 
consists of all named services with a name not in $F$ that belong to
the service family denoted by $t$.

The axioms of \SFA\ are given in 
Table~\ref{axioms-SFA}.%
\begin{table}[!t]
\caption{Axioms of \SFA}
\label{axioms-SFA}
{
\begin{eqntbl}
\begin{axcol}
u \sfcomp \emptysf = u                                 & \axiom{SFC1} \\
u \sfcomp v = v \sfcomp u                              & \axiom{SFC2} \\
(u \sfcomp v) \sfcomp w = u \sfcomp (v \sfcomp w)      & \axiom{SFC3} \\
f.z \sfcomp f.z' = f.\emptyserv                  & \axiom{SFC4}
\end{axcol}
\qquad
\begin{saxcol}
\encap{F}(\emptysf) = \emptysf                     & & \axiom{SFE1} \\
\encap{F}(f.z) = \emptysf & \mif f \in F       & \axiom{SFE2} \\
\encap{F}(f.z) = f.z    & \mif f \notin F    & \axiom{SFE3} \\
\multicolumn{2}{@{}l@{\quad}}
 {\encap{F}(u \sfcomp v) =
  \encap{F}(u) \sfcomp \encap{F}(v)}               & \axiom{SFE4}
\end{saxcol}
\end{eqntbl}
}
\end{table}
In this table, $f$ stands for an arbitrary focus from $\Foci$ and
$F$ stands for an arbitrary subset of $\Foci$.
These axioms simply formalize the informal explanation given
above.

For the set $\BAct$ of basic actions, we now take the set
$\set{f.m \where f \in \Foci, m \in \Meth}$.
Performing a basic action $f.m$ is taken as making a request to the
service named $f$ to process method $m$.

We combine \BTA\ with \SFA\ and extend the combination with the 
following operator:
\begin{itemize}
\item
the binary \emph{apply} operator
$\funct{\ph \sfapply \ph}{\Thr \x \ServFam}{\ServFam}$;
\end{itemize}
and the axioms given in Table~\ref{axioms-apply}.%
\begin{table}[!t]
\caption{Axioms for the apply operator}
\label{axioms-apply}
\begin{eqntbl}
\begin{saxcol}
\Stop  \sfapply u = u                                  & & \axiom{A1} \\
\DeadEnd \sfapply u = \emptysf                         & & \axiom{A2} \\
(\pcc{x}{f.m}{y}) \sfapply \encap{\set{f}}(u) = \emptysf
                                                       & & \axiom{A3} \\
(\pcc{x}{f.m}{y}) \sfapply (f.t \sfcomp \encap{\set{f}}(u)) =
x \sfapply (f.\derive{m}t \sfcomp \encap{\set{f}}(u))
                           & \mif \sreply{m}(t) = \True  & \axiom{A4} \\
(\pcc{x}{f.m}{y}) \sfapply (f.t \sfcomp \encap{\set{f}}(u)) =
y \sfapply (f.\derive{m}t \sfcomp \encap{\set{f}}(u))
                           & \mif \sreply{m}(t) = \False & \axiom{A5} \\
(\pcc{x}{f.m}{y}) \sfapply (f.t \sfcomp \encap{\set{f}}(u)) = \emptysf
                           & \mif \sreply{m}(t) = \Div   & \axiom{A6}
\end{saxcol}
\end{eqntbl}
\end{table}
In this table, $f$ stands for an arbitrary focus from $\Foci$, $m$ 
stands for an arbitrary method from $\Meth$, and $t$ stands for an 
arbitrary term of sort $\Serv$.
The axioms formalize the informal explanation given below and in 
addition stipulate what is the result of apply if inappropriate foci 
or methods are involved.
We use infix notation for the apply operator.

The service family denoted by a closed term of the form $t \sfapply t'$ 
is the service family that results from processing the method of each 
basic action performed by the thread denoted by $t$ by the service in 
the service family denoted by $t'$ with the focus of the basic action as 
its name if such a service exists.
When the method of a basic action performed by a thread is processed by
a service, the service changes in accordance with the method concerned
and the thread reduces to one of the two threads that it can possibly 
proceed with dependent on the reply value produced by the service.

In the case of the stack services described earlier in this section, the 
following two equations are simple examples of derivable equations: 
\begin{ldispl}
(\pcc{(\nns.\pop \bapf \Stop)}{\nns.\topeq{0}}{\Stop}) \sfapply
\nns.\NNS_{0 \sigma} = 
\nns.\NNS_\sigma\;,
\\
(\pcc{(\nns.\pop \bapf \Stop)}{\nns.\topeq{0}}{\Stop}) \sfapply
\nns.\NNS_{1 \sigma} = 
\nns.\NNS_{1 \sigma}\;.
\end{ldispl}%

\section{Hoare-Like Logic for \PGA\protect\footnotemark}
\label{sect-HL-PC}
\footnotetext
{The term ``Hoare-like logic'', which stands for ``logic like Hoare'' if 
 taken literally, is widely used since 1981 with the meaning ``logic like 
 Hoare logic'' and we conform to this usage.
}

In this section, we introduce a formal system for proving the partial 
correctness of instruction sequences as considered in \PGA.
Unlike segments of programs written in the high-level programming 
languages for which Hoare logics have been developed, segments of 
single-pass instruction sequences may have multiple entry points and 
multiple exit points.
Therefore, the asserted programs of the form $\assprog{P}{S}{Q}$ of 
Hoare logics fall short in the case of single-pass instruction 
sequences.
The formulas in the formal system introduced here will be called 
asserted instruction sequences.

We will look upon foci as (program) variables. 
This is justified by the fact that foci are names of objects that may be 
modified on execution of an instruction sequence.
The objects concerned are services.
What is assumed here with respect to services is the same as in 
Section~\ref{sect-TSI}.
This means that a signature $\Sig{\ServAlg}$ that includes specific 
sorts, constants and operators and a minimal $\Sig{\ServAlg}$-algebra 
$\ServAlg$ that satisfies specific conditions have been given.

In the formal system introduced here, classical first-order logic with
equality is used for pre- and post-conditions.
The particular choice of log\-i\-cal constants, connectives and 
quantifiers does not matter.
However, for convenience, it is assumed that the following is included:
(a)~the constants $\Ltrue$ (for truth) and $\Lfalse$ (for falsity), 
(b)~the connectives $\Lnot$ (for negation), $\Land$ (for conjunction), 
$\Lor$ (for disjunction), and $\Limpl$ (for implication), 
(c)~the quantifiers $\forall$ (for universal quantification) and 
$\exists$ (for existential quantification).

We write $\Lang{\ServAlg}$ for the many-sorted first-order language with 
equality over the signature $\Sig{\ServAlg}$ where free variables of 
sort $\Serv$ belong to the set $\Foci$.
Moreover, we write $\CTerm{\InSeq}$ for the set of all closed terms of 
sort $\InSeq$ in the case where the set 
$\set{f.m \where f \in \Foci, m \in \Meth}$ is taken as the set 
$\BInstr$ of basic instructions.

An \emph{asserted instruction sequence} is a formula of the form
$\assinseq{b}{P}{S}{e}{Q}$, where $S \in \CTerm{\InSeq}$,
$P,Q \in \Lang{\ServAlg}$, $b \in \Natpos$, and $e \in \Nat$.%
\footnote
{We write $\Natpos$ for the set $\set{n \in \Nat \where n > 0}$.}
The intuitive meaning of an asserted instruction sequence 
$\assinseq{b}{P}{S}{e}{Q}$ is as follows:
\begin{itemize}
\item
if $b \leq \len(S)$ and $e > 0$, the intuitive meaning is: 
\begin{quote}
if execution enters the instruction sequence segment $S$ at its $b$th 
instruction and $P$ holds when execution enters $S$, then either 
execution becomes inactive in $S$ or execution exits $S$ by going to 
the $e$th instruction following $S$ and $Q$ holds when execution exits 
$S$;
\end{quote}
\item
if $b \leq \len(S)$ and $e = 0$, the intuitive meaning is: 
\begin{quote}
if execution enters the instruction sequence segment $S$ at its $b$th 
instruction and $P$ holds when execution enters $S$, then either 
execu\-tion becomes inactive in $S$ or execution terminates in $S$ 
and $Q$ holds when execution terminates in \nolinebreak[2] $S$;%
\footnote
{Recall that execution becomes inactive if no more basic instructions
 are executed, but execution does not terminate.}
\end{quote}
\item
if $b > \len(S)$, an intuitive meaning is lacking.
\end{itemize}
For convenience, we did not exclude the case where $b > \len(S)$.
Instead, we made the choice that any asserted instruction sequence 
$\assinseq{b}{P}{S}{e}{Q}$ with $b > \len(S)$ does not hold 
(irrespective of the choice of $\ServAlg$).

Before we make precise what it means that an asserted instruction 
sequence holds in $\ServAlg$, we introduce some special terminology and 
notation.

In the setting of \PGA, what we mean by a \emph{state} is a function 
from a finite subset of $\Foci$ to the interpretation of sort $\Serv$ in 
$\ServAlg$.
Let $F \subset \Foci$ be such that $F$ is finite.
Then a \emph{state representing term for} $F$ \emph{with respect to} 
$\ServAlg$ is a closed term $t$ of sort $\ServFam$ for which, for all 
$f \in F$, $\encap{\set{f}}(t) = t$ does not hold in the free extension 
of $\ServAlg$ to a model of \SFA.
Notice that $\encap{\set{f}}(t) = t$ does not hold iff the 
interpretation of $t$ is a service family to which a service with name 
$f$ belongs.
Let $P \in \Lang{\ServAlg}$, and
let $F'$ be the set all foci that belong to the free variables of $P$.
Then a \emph{state representing term for} $P$ \emph{with respect to} 
$\ServAlg$ is a closed term $t$ of sort $\ServFam$ that is a state 
representing term for $F'$ with respect to $\ServAlg$.
Let $S \in \CTerm{\InSeq}$, and
let $F''$ be the set all foci that occur in $S$.
Then a \emph{state representing term for} $S$ \emph{with respect to} 
$\ServAlg$ is a closed term $t$ of sort $\ServFam$ that is a state 
representing term for $F''$ with respect to $\ServAlg$.

We write $P[t]$, where $t$ is a state representing term for $P$ with 
respect to $\ServAlg$, for $P$ with, for each $f \in \Foci$, all free 
occurrences of $f$ replaced by a closed term $t'$ of sort $\Serv$ such 
that $t = f.t' \sfcomp \encap{\set{f}}(t)$ holds in the free extension 
of $\ServAlg$ to a model of \SFA.
Thus, the interpretation of the term $t'$ replacing the free occurrences 
of $f$ is the service associated with $f$ in the state represented by 
$t$.
Notice that an equation between terms of sort $\ServFam$ holds in the 
free extension of $\ServAlg$ to a model of \SFA\ iff it is derivable 
from the axioms of \SFA.

We write $\extr{S}_{b,0}$ for $\extr{\fjmp{b} \conc S}$ and 
$\extr{S}_{b,e}$, where $e > 0$, for 
$\extr{\fjmp{b} \conc S \conc \sigma(e)}$ where, 
for each $n > 0$, $\sigma(n)$ is defined by induction on $n$ as follows: 
$\sigma(1) = \halt$ and $\sigma(n + 1) = \fjmp{0} \conc \sigma(n)$.
In the case where $b \leq \len(S) \leq \omega$ and $e > 0$, the thread 
denoted by $\extr{S}_{b,e}$ represents the behaviour that differs from 
the behaviour produced by the instruction sequence segment $S$ in 
isolation if execution enters the segment at its $b$th instruction only 
by terminating instead of becoming inactive if execution exits the 
segment by going to the $e$th instruction following it.
This adaptation of the behaviour is a technicality by which it is 
possible to obtain the state at the time that execution exits the 
segment by means of the apply operation $\sfapply$.

An asserted instruction sequence $\assinseq{b}{P}{S}{e}{Q}$ \emph{holds} 
in $\ServAlg$, written $\ServAlg \Sat \assinseq{b}{P}{S}{e}{Q}$, 
if $b \leq \len(S)$ and for all closed terms $t$ and $t'$ of sort 
$\ServFam$ that are state representing terms for $P$, $Q$, and $S$ with 
respect to $\ServAlg$:
\begin{center}
$\ServAlg \Sat P[t]$ implies 
$\ISM{\ServAlg} \Sat \extr{S}_{b,e'} \sfapply t = \emptysf$
for all $e' \in \Nat$ with $e \neq e'$ \\
and \\ 
$\ServAlg \Sat P[t]$ and 
$\ISM{\ServAlg} \Sat \extr{S}_{b,e} \sfapply t = t'$ imply 
$\ServAlg \Sat Q[t']$, 
\end{center}
where $\ISM{\ServAlg}$ is the model of the combination of \PGA, \BTA, 
and \SFA\ extended with the thread extraction operator, the apply 
operator, and the axioms for these operators such that the restrictions 
to the signatures of  \PGA, \BTA, and \SFA\ are the initial model of 
\PGA, the projective limit model of \BTA, and the free extension of 
$\ServAlg$ to a model of \SFA, respectively.
The existence of such a model follows from the fact that the signatures 
of \PGA, \BTA, and \SFA\ are disjoint by the amalgamation result about 
expansions presented as Theorem~6.1.1 in~\cite{Hod93a} (adapted to the 
many-sorted case).
The occurrences of $\ServAlg$ in the above definition can be replaced by 
$\ISM{\ServAlg}$.

Notice that 
for all $S \in \CTerm{\InSeq}$, $Q \in \Lang{\ServAlg}$, 
$b \in \Natpos$ with $b \leq \len(S)$, and $e \in \Nat$, 
$\ServAlg \Sat \assinseq{b}{\Lfalse}{S}{e}{Q}$.
However, there exist $S \in \CTerm{\InSeq}$, $P \in \Lang{\ServAlg}$, 
$b \in \Natpos$ with $b \leq \len(S)$, and $e \in \Nat$ such that 
$\ServAlg \not\Sat \assinseq{b}{P}{S}{e}{\Ltrue}$.
This is the case because, if execution enters the instruction sequence 
segment $S$ at its $b$th instruction and $P$ holds when execution enters 
$S$, then there may be no unique way in which execution exits $S$ and, 
if there is a unique way, it may be by going to another than the $e$th 
instruction following $S$.

We could have dealt with the above-mentioned non-uniqueness by 
supporting multiple exit points in asserted instruction sequences.
In that case, we would have asserted instruction sequences of the form 
$\assinseq{b}{P}{S}{e_1,\ldots,e_n}{Q}$ satisfying 
$\ServAlg \Sat \assinseq{b}{P}{S}{e_1,\ldots,e_n}{Q}$ iff
$\ServAlg \Sat \assinseq{b}{P}{S}{e_i}{Q}$ 
for all $i \in \set{1,\ldots,n}$.
This means that it is sufficient to add to the axioms and rules of 
inference of our Hoare-like logic (introduced below) the rules of 
inference corresponding to this equivalence.
These additional rules are such that nothing gets lost if 
$\assinseq{b}{P}{S}{e_1,\ldots,e_n}{Q}$ is simply considered a shorthand 
for the set 
$\set{\assinseq{b}{P}{S}{e_i}{Q} \where i \in \set{1,\ldots,n}}$ 
of asserted instruction sequences.

The axioms and rules of inference of our Hoare-like logic of asserted 
single-pass instruction sequences are given in 
Table~\ref{axioms-rules-hoare-logic}.
\begin{table}[!p]
\caption{Hoare-Like Logic of Asserted Single-Pass Instruction Sequences}
\label{axioms-rules-hoare-logic}
\tblsize \centering \vspace*{-2ex}
\begin{minipage}[t]{4.675in} 
\hrulefill 
\\[.3ex]
\textsc{Basic Instruction Axioms:}
\vspace*{-1.9ex}
\begin{lldispl}
\Axiom{A1}
      {\assinseq{1}{\sreply{m}(f) \neq \Div \Land 
                    P\subst{\derive{m}(f)}{f}}{f.m}{1}{P}}
\\
\Axiom{A2}{\assinseq{1}{\sreply{m}(f) = \Div}{f.m}{0}{\Lfalse}}
\vspace*{-1.9ex}
\end{lldispl}%
\textsc{Positive Test Instruction Axioms:}
\vspace*{-1.9ex}
\begin{lldispl}
\Axiom{A3}
      {\assinseq{1}{\sreply{m}(f) = \True \Land 
                    P\subst{\derive{m}(f)}{f}}{\ptst{f.m}}{1}{P}}
\\
\Axiom{A4}
      {\assinseq{1}{\sreply{m}(f) = \False \Land 
                    P\subst{\derive{m}(f)}{f}}{\ptst{f.m}}{2}{P}}
\\
\Axiom{A5}{\assinseq{1}{\sreply{m}(f) = \Div}{\ptst{f.m}}{0}{\Lfalse}}
\vspace*{-1.9ex}
\end{lldispl}%
\textsc{Negative Test Instruction Axioms:}
\vspace*{-1.9ex}
\begin{lldispl}
\Axiom{A6}
      {\assinseq{1}{\sreply{m}(f) = \True \Land 
                    P\subst{\derive{m}(f)}{f}}{\ntst{f.m}}{2}{P}}
\\
\Axiom{A7}
      {\assinseq{1}{\sreply{m}(f) = \False \Land 
                    P\subst{\derive{m}(f)}{f}}{\ntst{f.m}}{1}{P}}
\\
\Axiom{A8}{\assinseq{1}{\sreply{m}(f) = \Div}{\ntst{f.m}}{0}{\Lfalse}}
\vspace*{-1.9ex}
\end{lldispl}%
\textsc{Forward Jump Instruction Axioms:}
\vspace*{-1.9ex}
\begin{lldispl}
\Axiom{A9}{\assinseq{1}{P}{\fjmp{i{+}1}}{i{+}1}{P}}
\qquad
\Axiom{A10}{\assinseq{1}{\Ltrue}{\fjmp{0}}{0}{\Lfalse}}
\vspace*{-1.9ex}
\end{lldispl}%
\textsc{Termination Instruction Axiom:}
\vspace*{-1.9ex}
\begin{lldispl}
\Axiom{A11}{\assinseq{1}{P}{\halt}{0}{P}}
\vspace*{-1.9ex}
\end{lldispl}%
\textsc{Concatenation Rules:}
\vspace*{-1.9ex}
\begin{lldispl}
\RuleC{R1}{\assinseq{b}{P}{S_1}{i}{Q}, \; \assinseq{i}{Q}{S_2}{e}{R}}
      {\assinseq{b}{P}{S_1 \conc S_2}{e}{R}}{i>0}
\\
\RuleC{R2}{\assinseq{b}{P}{S_1}{e{+}\len(S_2)}{Q}}
      {\assinseq{b}{P}{S_1 \conc S_2}{e}{Q}}{e > 0}
\qquad
\Rule{R3}{\assinseq{b}{P}{S_1}{0}{Q}}{\assinseq{b}{P}{S_1 \conc S_2}{0}{Q}}
\\
\Rule{R4}{\assinseq{b}{P}{S_2}{e}{Q}}
     {\assinseq{b{+}\len(S_1)}{P}{S_1 \conc S_2}{e}{Q}}
\vspace*{-1.9ex}
\end{lldispl}%
\textsc{Repetition Rule} (for each $k,n > 0$ with $k \leq n$):
\vspace*{-1.9ex}
\begin{lldispl}
\RuleS{R5}
         {\begin{array}[b]{@{}c@{}}
          \assinseq{b_1}{P_1}{S\rep}{0}{Q_1}, \ldots, 
          \assinseq{b_n}{P_n}{S\rep}{0}{Q_n}
           \Entr \assinseq{b_1}{P_1}{S \conc S\rep}{0}{Q_1}
          \\ \vdots \\
          \assinseq{b_1}{P_1}{S\rep}{0}{Q_1}, \ldots, 
          \assinseq{b_n}{P_n}{S\rep}{0}{Q_n}
           \Entr \assinseq{b_n}{P_n}{S \conc S\rep}{0}{Q_n}
          \end{array}}
     {\assinseq{b_k}{P_k}{S\rep}{0}{Q_k}}{4.5}
\vspace*{-1.9ex}
\end{lldispl}%
\textsc{Alternatives Rule:}
\vspace*{-1.9ex}
\begin{lldispl}
\Rule{R6}{\assinseq{b}{P}{S}{e}{R},\; \assinseq{b}{Q}{S}{e}{R}}
     {\assinseq{b}{P \Lor Q}{S}{e}{R}}
\vspace*{-1.9ex}
\end{lldispl}%
\textsc{Invariance Rule:}
\vspace*{-1.9ex}
\begin{lldispl}
\RuleC{R7}{\assinseq{b}{P}{S}{e}{Q}}
      {\assinseq{b}{P \Land R}{S}{e}{Q \Land R}}
      {\vars(R) \inter \vars(S) = \emptyset}
\vspace*{-1.9ex}
\end{lldispl}%
\textsc{Elimination Rule:}
\vspace*{-1.9ex}
\begin{lldispl}
\RuleC{R8}{\assinseq{b}{P}{S}{e}{Q}}
      {\assinseq{b}{\Exists{x}{P}}{S}{e}{Q}}
      {\set{x} \inter (\vars(S) \union \vars(Q)) = \emptyset}
\vspace*{-1.9ex}
\end{lldispl}%
\textsc{Substitution Rule:}
\vspace*{-1.9ex}
\begin{lldispl}
\RuleC{R9}{\assinseq{b}{P}{S}{e}{Q}}
      {\assinseq{b}{P\subst{y}{x}}{S}{e}{Q\subst{y}{x}}}
      {\set{x} \inter \vars(S) = \emptyset,\; 
       \set{y} \inter \vars(S) = \emptyset}
\vspace*{-1.9ex}
\end{lldispl}%
\textsc{Consequence Rule:}
\vspace*{-1.9ex}
\begin{lldispl}
\Rule{R10}{P \Limpl P',\; \assinseq{b}{P'}{S}{e}{Q'},\; Q' \Limpl Q}
     {\assinseq{b}{P}{S}{e}{Q}}
\vspace*{-2.7ex}
\end{lldispl}%
\hrulefill 
\end{minipage}
\end{table}
In this table, 
$S,S_1,S_2$ stand for arbitrary closed terms from $\CTerm{\InSeq}$, 
$P,P',P_1,P_2,\ldots$, $Q,Q',Q_1,Q_2,\ldots$, and $R$ stand for 
arbitrary formulas from $\Lang{\ServAlg}$, 
$b,b_1,b_2,\ldots$ stand for arbitrary positive natural numbers, 
$e,i$ stand for arbitrary natural numbers, 
$x,y$ stand for arbitrary variables of some sort in $\Sig{\ServAlg}$, 
$f$ stands for an arbitrary focus from $\Foci$, and 
$m$ stands for an arbitrary method from $\Meth$. 
Moreover, 
$\vars(P)$ denotes the set all foci that belong to the free variables of 
$P$ and $\vars(S)$ denotes the set of all foci that occur in $S$. 
We write $\Psi \Entr \phi$, where $\Psi$ is a finite set of asserted 
instruction sequences and $\phi$ is an asserted instruction sequence, 
for provability of $\phi$ from $\Psi$ without applications of the 
repetition \nolinebreak[2] rule (R5).

The axioms concern the smallest instruction sequence segments, namely
single instructions. 
Axioms A1--A8 are similar to the assignment axiom found in most Hoare 
logics.
They are somewhat more complicated than the assignment axiom because 
they concern instructions that may cause execution to become inactive 
and, in case of axioms A3--A8, instructions that have two exit points.
Axioms A9--A11, which concern jump instructions and the termination
instruction, are very simple and speak for themselves.

Concatenation needs four rules because instruction sequence segments may 
be prefixed or suffixed by redundant instruction sequence segments in 
several ways.
Rule R1 concerns the obvious case, namely the case where execution 
enters the whole by entering the first instruction sequence segment and 
execution exits the whole by exiting the second instruction sequence 
segment.
Rule R2 concerns the case where execution exits the whole by exiting the
first instruction sequence segment.
Rule R3 concerns the case where execution becomes inactive or terminates
in the whole by doing so in the first instruction sequence segment.
Rule R4 concerns the case where execution enters the whole by entering 
the second instruction sequence~segment.

The repetition rule (rule R5) is reminiscent of the recursion rule found 
in Hoare logics for high-level programming languages that covers calls 
of (parameterless) recursive procedures (see e.g.~\cite{Apt81a}).
This rule is actually a rule schema: there is an instance of this rule 
for each $k,n > 0$ with $k \leq n$.
In many cases, the instance for $k = 1$ and $n = 1$ suffices.
The need for the rules R6--R9 is not clear at first sight, but without
them the presented formal system would be incomplete.
Although these rules do not explicitly deal with repetition, they would 
not be needed for completeness in the absence of repetition.

The consequence rule (rule R10) is found in one form or another in 
all Hoare logics and Hoare-like logics.
This rule allows to make use of formulas from $\Lang{\ServAlg}$ that 
hold in $\ServAlg$ to strengthen pre-conditions and weaken 
post-conditions.

Because there is no rule of inference to deal with nested repetitions,
it seems at first sight that we cannot have a completeness result for
the presented Hoare-like logic.
However, a closer look at this matter yields something different.
The crux is that the following rule of inference is derivable from
rules R3 and R5:
\begin{ldispl}
\RuleU{\assinseq{b}{P}{S}{0}{Q}}{\assinseq{b}{P}{S\rep}{0}{Q}}\;.
\end{ldispl}%
We have the following result:
\begin{theorem}
\label{theorem-nested-rep}
Let $\Th{\ServAlg}$ be the set of all formulas of $\Lang{\ServAlg}$ that 
hold in $\ServAlg$.
Then, for each $S \in \CTerm{\InSeq}$, $P,Q \in \Lang{\ServAlg}$, 
and $b \in \Natpos$, $\ServAlg \Sat \assinseq{b}{P}{S}{0}{Q}$ only if
there exists an $S' \in \CTerm{\InSeq}$ in which the repetition operator
occurs at most once such that 
(a)~$\ServAlg \Sat \assinseq{b}{P}{S'}{0}{Q}$ and 
(b)~$\Th{\ServAlg} \Ent \assinseq{b}{P}{S'}{0}{Q}$ implies
$\Th{\ServAlg} \Ent \assinseq{b}{P}{S}{0}{Q}$.
\end{theorem}
\begin{proof}
Let $S \in \CTerm{\InSeq}$ be such that the repetition operator
occurs at least once in $S$.
Then the following properties follow directly from the definitions 
involved ((1) and (2)) and the presented Hoare-like logic
((3) and (4)):
\begin{enumerate}
\item[(1)]
$\ServAlg \Sat \assinseq{b}{P}{S \conc T}{0}{Q}$ implies
$\ServAlg \Sat \assinseq{b}{P}{S}{0}{Q}$;
\item[(2)]
$\ServAlg \Sat \assinseq{b}{P}{S\rep}{0}{Q}$ implies
$\ServAlg \Sat \assinseq{b}{P}{S}{0}{Q}$;
\item[(3)]
$\Th{\ServAlg} \Ent \assinseq{b}{P}{S}{0}{Q}$ implies
$\Th{\ServAlg} \Ent \assinseq{b}{P}{S \conc T}{0}{Q}$;
\item[(4)]
$\Th{\ServAlg} \Ent \assinseq{b}{P}{S}{0}{Q}$ implies
$\Th{\ServAlg} \Ent \assinseq{b}{P}{S\rep}{0}{Q}$.
\end{enumerate}
Using these properties, the theorem is easily proved by induction on the 
number of occurrences of the repetition operator in $S$.
\qed
\end{proof}
As a corollary of Theorem~\ref{theorem-nested-rep} we have that a
completeness result for the set of all closed \PGA\ terms of sort 
$\InSeq$ in which the repetition operator occurs at most once entails a
completeness result for the set of all closed \PGA\ terms of sort 
$\InSeq$.

\section{Example}
\label{sect-example}

In this section, we give an example of the use of the Hoare-like logic 
of asserted single-pass instruction sequences presented in 
Section~\ref{sect-HL-PC}.
The example has only been chosen because it is simple and shows 
applications of most axioms and rules of inference of this Hoare-like 
logic (including R6 and R8).

For $\ServAlg$, we take an algebra of services that make up unbounded
natural number counters.
Each natural number counter service is able to process methods to 
increment the content of the counter by one ($\incr$), to decrement the 
content of the counter by one ($\decr$), and to test whether the content 
of the counter is zero ($\iszero$).
The derived service and service reply operations for these methods are as
to be expected.
$\Sig{\ServAlg}$ includes the sort $\NN$ of natural numbers, the constant 
$\const{0}{\NN}$, and the unary operators $\funct{\nnsucc}{\NN}{\NN}$, 
$\funct{\nnpred}{\NN}{\NN}$, and $\funct{\nnc}{\NN}{\Serv}$.
The interpretation of $\NN$, $0$, $\nnsucc$, and $\nnpred$ are as to be 
expected.
The interpretation of $\nnc$ is the function that maps each natural 
number $n$ to the service that makes up a counter whose content is $n$.

We claim that the closed \PGA\ term
$(\ntst{\nncvar.\iszero} \conc \fjmp{2} \conc \halt \conc
 \nncvar.\decr)\rep$ 
denotes an instruction sequence for setting the counter made up by 
service $\nncvar$ to zero.
That is, we claim
$\assinseq{1}{\Ltrue}
          {(\ntst{\nncvar.\iszero} \conc \fjmp{2} \conc \halt \conc
            \nncvar.\decr)\rep}
          {0}{\nncvar = \nnc(0)}$.
We prove this by means of the axioms and rules of inference given in
Table~\ref{axioms-rules-hoare-logic}.

It is sufficient to prove
\begin{enumerate}
\item[(1)]
$\assinseq{1}{\nncvar = \nnc(0) \Lor \nncvar = \nnc(n + 1)}
          {(\ntst{\nncvar.\iszero} \conc \fjmp{2} \conc \halt \conc
            \nncvar.\decr)\rep}
          {0}{\nncvar = \nnc(0)}$
\end{enumerate}
because the claim follows from (1) by R8 and R10.

\noindent
First, we prove 
$\assinseq{1}{\nncvar = \nnc(0)}
          {\ntst{\nncvar.\iszero} \conc \fjmp{2} \conc \halt \conc
           \nncvar.\decr}
          {0}{\nncvar = \nnc(0)}$:
\begin{enumerate}
\item[(2)]
$\assinseq{1}{\nncvar = \nnc(0)}{\ntst{\nncvar.\iszero}}
          {2}{\nncvar = \nnc(0)}$ \\ \hsp{2}
by A6;
\item[(3)]
$\assinseq{1}{\nncvar = \nnc(0)}
          {\ntst{\nncvar.\iszero} \conc \fjmp{2}}
          {1}{\nncvar = \nnc(0)}$ \\ \hsp{2}
from~(2) by A9 and R2;
\item[(4)]
$\assinseq{1}{\nncvar = \nnc(0)}
          {\ntst{\nncvar.\iszero} \conc \fjmp{2} \conc \halt}
          {0}{\nncvar = \nnc(0)}$ \\ \hsp{2}
from~(3) by A11 and R1;
\item[(5)]
$\assinseq{1}{\nncvar = \nnc(0)}
          {\ntst{\nncvar.\iszero} \conc \fjmp{2} \conc \halt \conc
           \nncvar.\decr}
          {0}{\nncvar = \nnc(0)}$ \\ \hsp{2}
from~(4) by A1~and R3.
\end{enumerate}
Next, we prove
$\assinseq{1}{\nncvar = \nnc(n + 1)}
          {\ntst{\nncvar.\iszero} \conc \fjmp{2} \conc \halt \conc
           \nncvar.\decr}
          {0}{\nncvar = \nnc(n)}$:
\begin{enumerate}
\item[(6)]
$\assinseq{1}{\nncvar = \nnc(n + 1)}{\ntst{\nncvar.\iszero}}
          {1}{\nncvar = \nnc(n + 1)}$ \\ \hsp{2}
by A6;
\item[(7)]
$\assinseq{1}{\nncvar = \nnc(n + 1)}
          {\ntst{\nncvar.\iszero} \conc \fjmp{2}}
          {2}{\nncvar = \nnc(n + 1)}$ \\ \hsp{2}
from~(6) by A9 and R1;
\item[(8)]
$\assinseq{1}{\nncvar = \nnc(n + 1)}
          {\ntst{\nncvar.\iszero} \conc \fjmp{2} \conc \halt}
          {1}{\nncvar = \nnc(n + 1)}$ \\ \hsp{2}
from~(7) by A11 and R2;
\item[(9)]
$\assinseq{1}{\nncvar = \nnc(n + 1)}
          {\ntst{\nncvar.\iszero} \conc \fjmp{2} \conc \halt \conc
           \nncvar.\decr}
          {0}{\nncvar = \nnc(n)}$ \\ \hsp{2}
from~(8) by A1, R10 and R1.
\end{enumerate}
Assuming (1), we prove 
\begin{enumerate}
\item[]
$\lassinseq{1}{\nncvar = \nnc(0) \Lor \nncvar = \nnc(n + 1)}
          {\ntst{\nncvar.\iszero} \conc \fjmp{2} \conc \halt \conc
           \nncvar.\decr \conc
           (\ntst{\nncvar.\iszero} \conc \fjmp{2} \conc \halt \conc
            \nncvar.\decr)\rep}
          {0}{\nncvar = \nnc(0)}$:
\end{enumerate}
\begin{enumerate}
\item[(a)]
$\lassinseq{1}{\nncvar = \nnc(0)}
          {\ntst{\nncvar.\iszero} \conc \fjmp{2} \conc \halt \conc
            \nncvar.\decr \conc 
           (\ntst{\nncvar.\iszero} \conc \fjmp{2} \conc \halt \conc
            \nncvar.\decr)\rep}
          {0}{\nncvar = \nnc(0)}$ \\ \hsp{2}
from~(5) by R3;
\item[(b)]
$\lassinseq{1}{\nncvar = \nnc(n + 1)}
          {\ntst{\nncvar.\iszero} \conc \fjmp{2} \conc \halt \conc
           \nncvar.\decr \conc
           (\ntst{\nncvar.\iszero} \conc \fjmp{2} \conc \halt \conc
            \nncvar.\decr)\rep}
          {0}{\nncvar = \nnc(0)}$ \\ \hsp{2}
from~(9) by R1;
\item[(c)]
$\lassinseq{1}{\nncvar = \nnc(0) \Lor \nncvar = \nnc(n + 1)}
          {\ntst{\nncvar.\iszero} \conc \fjmp{2} \conc \halt \conc
           \nncvar.\decr \conc
           (\ntst{\nncvar.\iszero} \conc \fjmp{2} \conc \halt \conc
            \nncvar.\decr)\rep}
          {0}{\nncvar = \nnc(0)}$ \\ \hsp{2}
from~(a) and~(b) by R6.
\end{enumerate} 
Because (c) has been derived assuming (1), (1) now follows by R5.

The example given above illustrates that proving instruction sequences
correct can be quite tedious, even in a simple case. 
This can be largely attributed to the fact that instruction sequences 
do not need to be structured programs and not to the particular 
Hoare-like logic used.
A verification condition generator and a proof assistant are anyhow 
indispensable when proving realistic instruction sequences correct.

\section{Soundness and Completeness}
\label{sect-sound-complete}

This section is concerned with the soundness and completeness of the 
Hoare-like logic of asserted single-pass instruction sequences 
presented in Section~\ref{sect-HL-PC}.
It was assumed in Section~\ref{sect-HL-PC} that a signature 
$\Sig{\ServAlg}$ that includes specific sorts, constants and operators 
and a minimal $\Sig{\ServAlg}$-algebra $\ServAlg$ that satisfies 
specific conditions had been given.
In this section, we intend to establish soundness and completeness for 
all algebras that could have been given.
It is useful to introduce a name for these algebras: 
\emph{service algebras}.

In this section, we write $\Th{\ServAlg}$, where $\ServAlg$ is a service 
algebra, for the set of all formulas of $\Lang{\ServAlg}$ that hold in 
$\ServAlg$.

The proof of the soundness theorem for the presented Hoare-like logic 
given below (Theorem~\ref{theorem-soundness}) will make use of the 
following two lemmas.
Recall that $\Entr$ stands for provability without applications of the 
repetition rule.
\begin{lemma}
\label{lemma-replacement}
Let $\ServAlg$ be a service algebra, and
let $k,n \in \Natpos$ be such that $k \leq n$.
Then, for each $S,S' \in \CTerm{\InSeq}$, 
$P_1,\ldots,P_n,Q_1,\ldots,Q_n \in \Lang{\ServAlg}$, and 
$b_1,\ldots,b_n \in \Natpos$, if
$\assinseq{b_1}{P_1}{S\rep}{0}{Q_1}, \ldots, 
 \assinseq{b_n}{P_n}{S\rep}{0}{Q_n} \Entr 
 \assinseq{b_k}{P_k}{S \conc S\rep}{0}{Q_k}$
then
$\assinseq{b_1}{P_1}{S'}{0}{Q_1}, \ldots, 
 \assinseq{b_n}{P_n}{S'}{0}{Q_n} \Entr 
 \assinseq{b_k}{P_k}{S \conc S'}{0}{Q_k}$.
\end{lemma}
\begin{proof}
This is easily proved by induction on the length of proofs, case 
distinction on the axiom applied in the basis step, and case distinction 
on the rule of inference last applied in the inductive step.
\qed
\end{proof}
An  important corollary of Lemma~\ref{lemma-replacement} is that, 
for all $i \in \Nat$ and $k \in \Natpos$ with $k \leq n$,
$\assinseq{b_1}{P_1}{S\rep}{0}{Q_1}, \ldots,
 \assinseq{b_n}{P_n}{S\rep}{0}{Q_n} \Entr 
 \assinseq{b_k}{P_k}{S \conc S\rep}{0}{Q_k}$
only if
$\assinseq{b_1}{P_1}{S^i}{0}{Q_1}, \ldots,
 \assinseq{b_n}{P_n}{S^i}{0}{Q_n} \Entr 
 \assinseq{b_k}{P_k}{S^{i+1}}{0}{Q_k}$.
\begin{lemma}
\label{lemma-soundness}
For each service algebra $\ServAlg$, set of asserted instruction 
sequences $\Psi$, and asserted instruction sequence $\phi$, 
$\Th{\ServAlg} \union \Psi \Entr \phi$ only if $\ServAlg \Sat \psi$ for
all $\psi \in \Psi$ implies $\ServAlg \Sat \phi$.
\end{lemma}
\begin{proof}
This is easily proved by induction on the length of proofs, case 
distinction on the axiom applied in the basis step, and case distinction 
on the rule of inference last applied in the inductive step.
\qed
\end{proof}
Lemma~\ref{lemma-soundness} expresses that, if the repetition rule is 
dropped, the axioms and inference rules of the presented Hoare-like 
logic are strongly sound.

The following theorem is the soundness theorem for the presented 
Hoare-like logic.
\begin{theorem}
\label{theorem-soundness}
For each service algebra $\ServAlg$ and asserted instruction sequence 
$\phi$, $\Th{\ServAlg} \Ent \phi$ implies $\ServAlg \Sat \phi$. 
\end{theorem}
\begin{proof}
This is proved by induction on the length of proofs, case distinction on 
the axiom applied in the basis step, and case distinction on the rule of 
inference last applied in the inductive step.
The only difficult case is the repetition rule (R5).
We will only outline the proof for this case.

The following properties follow directly from the definition of 
$\ISM{\ServAlg}$:
\begin{enumerate}
\item[(1)]
$\ISM{\ServAlg} \Sat \extr{S^0 \conc \fjmp{0}^b}_{b,0} \sfapply 
 t = \emptysf$;
\item[(2)]
$\ISM{\ServAlg} \Sat \extr{S\rep}_{b,0} \sfapply t = t'$ iff 
there exists an $j > 0$ such that:
\begin{itemize}
\item[]
for all $k \geq j$,
$\ISM{\ServAlg} \Sat \extr{S^k \conc \fjmp{0}^b}_{b,0} \sfapply t = t'$,
\item[]
for all $k < j$,
$\ISM{\ServAlg} \Sat \extr{S^k \conc \fjmp{0}^b}_{b,0} \sfapply 
 t = \emptysf$.
\end{itemize}
\end{enumerate}
These properties could be largely proved in a formal way if the combined 
algebraic theory of $\ISM{\ServAlg}$ developed in 
Sections~\ref{sect-PGA-and-BTA} and~\ref{sect-TSI} would be extended 
with projection operators and axioms for them as in~\cite{BM09k}.

The following properties follow directly from properties~(1) and~(2):
\begin{enumerate}
\item[(a)]
$\ServAlg \Sat \assinseq{b}{P}{S^0 \conc \fjmp{0}^b}{0}{Q}$; 
\item[(b)]
$\ServAlg \Sat \assinseq{b}{P}{S\rep}{0}{Q}$ iff,
for all $i \geq 0$, 
$\ServAlg \Sat \assinseq{b}{P}{S^i \conc \fjmp{0}^b}{0}{Q}$.
\end{enumerate}

Let $k,n \in \Natpos$ be such that $k \leq n$, and 
let $S \in \CTerm{\InSeq}$, 
$P_1,\ldots,P_n,Q_1,\ldots,Q_n \in \Lang{\ServAlg}$, and 
$b_1,\ldots,b_n \in \Natpos$.
Then, from the hypotheses of R5 and Lemmas~\ref{lemma-replacement}
and~\ref{lemma-soundness}, it follows immediately that, for all 
$i \geq 0$, 
\mbox{$\ServAlg \Sat \assinseq{b_1}{P_1}{S^i \conc \fjmp{0}^b}{0}{Q_1}$} 
and \ldots\ and
$\ServAlg \Sat \assinseq{b_n}{P_n}{S^i \conc \fjmp{0}^b}{0}{Q_n}$ 
implies 
$\ServAlg \Sat \assinseq{b_k}{P_k}{S^{i+1} \conc \fjmp{0}^b}{0}{Q_k}$.
From this and property~(a), it follows by induction on $i$ that, 
for all $i \geq 0$, 
$\ServAlg \Sat \assinseq{b_k}{P_k}{S^i \conc \fjmp{0}^b}{0}{Q_k}$.
From this and property~(b), it follows immediately that
$\ServAlg \Sat \assinseq{b_k}{P_k}{S\rep}{0}{Q_k}$.
This completes the proof for the case of the repetition rule.
\qed
\end{proof}
The line of the proof of Theorem~\ref{theorem-soundness} for the case 
that the rule of inference last applied is R5 is reminiscent of the line 
of the soundness proof in~\cite{Cla79a} for the case that the rule of 
inference last applied is the recursion rule for calls of recursive 
procedures.
In the proof of Theorem~\ref{theorem-soundness}, $S^i \conc \fjmp{0}^b$ 
is used instead of $S^i$ to guarantee that $b$ is never greater than the 
length of the approximations of $S\rep$.

There is a problem with establishing completeness for all service 
algebras.
In the completeness proof, it has to be assumed that, for each service 
algebra $\ServAlg$, necessary intermediate conditions can be expressed 
in $\Lang{\ServAlg}$.
Therefore, completeness will only be established for all service 
algebras that are sufficiently expressive.

Let $\ServAlg$ be a service algebra, and let $S \in \CTerm{\InSeq}$, 
$P,Q \in \Lang{\ServAlg}$, $b \in \Natpos$ and $e \in \Nat$.
Then $Q$ \emph{expresses the strongest post-condition of $P$ and $S$ for
$b$ and $e$ on $\ServAlg$} if 
$\ServAlg \Sat \assinseq{b}{P}{S}{e}{\Ltrue}$ and,
for each state representing term $t'$ for $P$, $Q$, and $S$ with respect 
to $\ServAlg$, $\ServAlg \Sat Q[t']$ iff 
there exists a state representing term $t$ for $P$, $Q$, and $S$ with 
respect to $\ServAlg$ such that $\ServAlg \Sat P[t]$ and 
$\ISM{\ServAlg} \Sat \extr{S}_{b,e} \sfapply t = t'$.

Let $\ServAlg$ be a service algebra.
Then the language $\Lang{\ServAlg}$ is \emph{expressive for 
$\CTerm{\InSeq}$ on $\ServAlg$} if, for each $S \in \CTerm{\InSeq}$, 
$P \in \Lang{\ServAlg}$, $b \in \Natpos$, and $e \in \Nat$ with 
$\ServAlg \Sat \assinseq{b}{P}{S}{e}{\Ltrue}$, there exists a 
$Q \in \Lang{\ServAlg}$ such that $Q$ expresses the strongest 
post-condition of $P$ and $S$ for $b$ and $e$ on $\ServAlg$.

In the above definitions, $\ServAlg \Sat \assinseq{b}{P}{S}{e}{\Ltrue}$
is used to express that there exists a post-condition of $P$ and $S$ for
$b$ and $e$ on $\ServAlg$. 

The following remarks about the existence of strongest post-conditions 
may be useful for a clear understanding of the matter.
For each $S \in \CTerm{\InSeq}$, $P \in \Lang{\ServAlg}$, and
$b \in \Natpos$, one of the following is the case regarding the 
existence of a strongest post-condition:
\begin{enumerate}
\item[(1)]
there is no $e \in \Nat$ for which there exists a strongest 
post-condition of $P$ and $S$ for $b$ and $e$;
\item[(2)]
there is exactly one $e \in \Nat$ for which there exists a strongest 
post-condition of $P$ and $S$ for $b$ and $e$ and the strongest 
post-condition concerned is not equivalent to $\Lfalse$;
\item[(3)]
there is more than one $e \in \Nat$ for which there exists a strongest 
post-condi\-tion of $P$ and $S$ for $b$ and $e$ and the strongest 
post-condition concerned is equivalent to $\Lfalse$.
\end{enumerate}
We say that execution is convergent in $S$ if it does not become 
inactive in $S$.
Terminating in $S$ is one way in which execution may be convergent in 
$S$, exiting $S$ by going to the $e$th instruction following $S$ is 
another way in which execution may be convergent in $S$, and exiting $S$ 
by going to the $e'$th instruction following $S$, where $e' \neq e$, is 
still another way in which execution may be convergent in~$S$.
Now, (1) is the case if there is more than one way in which execution 
may be convergent in $S$, (2) is the case if there is exactly one 
way in which execution may be convergent in $S$, and (3) is the case 
if there is no way in which execution may be convergent in $S$.

The proof of the completeness theorem for the presented Hoare-like 
logic given below (Theorem~\ref{theorem-completeness}) will make use of 
the following four lemmas.
\begin{lemma}
\label{lemma-repetition}
Let $\ServAlg$ be a service algebra.
Then, for each $S \in \CTerm{\InSeq}$, $P,Q \in \Lang{\ServAlg}$, 
$b \in \Natpos$, and $e \in \Nat$, 
$\ServAlg \Sat \assinseq{b}{P}{S\rep}{e}{Q}$ only if $e = 0$.
\end{lemma}
\begin{proof}
This is proved by distinguishing two cases: the repetition operator does 
not occur in $S$ and the repetition operator occurs in $S$.
The former case is easily proved by induction on $\len(S)$.
The latter case follows directly from the following corollary of the 
proof of Lemma~2.6 from~\cite{BM12b}: 
for each $S \in \CTerm{\InSeq}$ in which the repetition operator occurs, 
there exists an $S'$ in which the repetition operator does not occur 
such that $\extr{S}_{b,e} = \extr{{S'}\rep}_{b,e}$.
\qed
\end{proof}
Lemma~\ref{lemma-repetition} tells us that execution never exits an 
instruction sequence segment of the form $S\rep$.

The following lemma expresses that the axioms and rules of inference of 
the presented Hoare-like logic are complete for all instruction sequence 
segments of the form $S\rep$ only if they are complete for all 
instruction sequence segments.
\begin{lemma}
\label{lemma-completeness}
\sloppy
Let $\ServAlg$ be a service algebra such that $\Lang{\ServAlg}$ is 
expressive for $\CTerm{\InSeq}$ on $\ServAlg$.
Assume that, for each $S \in \CTerm{\InSeq}$, $P,Q \in \Lang{\ServAlg}$, 
$b \in \Natpos$, and $e \in \Nat$, 
$\ServAlg \Sat \assinseq{b}{P}{S\rep}{e}{Q}$ implies 
$\Th{\ServAlg} \Ent \assinseq{b}{P}{S\rep}{e}{Q}$.
Then, for each $S \in \CTerm{\InSeq}$, \mbox{$P,Q \in \Lang{\ServAlg}$}, 
$b \in \Natpos$, and $e \in \Nat$, 
$\ServAlg \Sat \assinseq{b}{P}{S}{e}{Q}$ implies 
$\Th{\ServAlg} \Ent \assinseq{b}{P}{S}{e}{Q}$.
\end{lemma}
\begin{proof}
This is proved by induction on the structure of $S$.
The cases that $S$ is a single instructions follow, with the exception 
of the termination instruction after a case distinction, directly from 
one of the axioms (A1--A11) and the con\-sequence rule (R10).
The case that $S$ is of the form ${S'}\rep$ follows immediately
from the assumption.
What is left is the case that $S$ is of the form $S_1 \conc S_2$.

If $\ServAlg \Sat \assinseq{b}{P}{S_1 \conc S_2}{e}{Q}$, then it follows 
from the definitions involved that:
\begin{enumerate}
\item[(1)]
if $b \leq \len(S_1)$:\,\, for some $n > 0$, 
there exist $P_1,R_1,\ldots,P_n,R_n \in \Lang{\ServAlg}$ and 
$i_1,\ldots,i_n \in \Natpos$ such that 
$\ServAlg \Sat P \Limpl P_1 \Lor \ldots \Lor P_n$ and,
for each $j$ with $1 \leq j \leq n$, 
$R_j$ expresses the strongest post-condition of $P_j$ and $S_1$ for $b$ 
and $i_j$ on~$\ServAlg$ and one of the following is the case:
\begin{enumerate}
\item[(a)]
$1 \leq i_j \leq \len(S_2)$, \\
$\ServAlg \Sat \assinseq{b}{P_j}{S_1}{i_j}{R_j}$, and
$\ServAlg \Sat \assinseq{i_j}{R_j}{S_2}{e}{Q}$;
\item[(b)]
$i_j = \len(S_2) + e$, $e > 0$, \\
$\ServAlg \Sat \assinseq{b}{P_j}{S_1}{i_j}{R_j}$, and
$\ServAlg \Sat R_j \Limpl Q$;
\item[(c)]
$i_j = 0$, $e = 0$, \\
$\ServAlg \Sat \assinseq{b}{P_j}{S_1}{i_j}{R_j}$, and
$\ServAlg \Sat R_j \Limpl Q$;
\end{enumerate}
\item[(2)]
if $b > \len(S_1)$:\,\,
$\ServAlg \Sat \assinseq{b - \len(S_1)}{P}{S_2}{e}{Q}$.
\end{enumerate}
Case~(1) is proved by distinguishing two subcases: the repetition 
operator does not occur in $S_1$ and the repetition operator occurs in 
$S_1$.
The former subcase is easily proved by induction on $\len(S_1)$.
The latter subcase follows directly from the above-mentioned corollary 
of the proof of Lemma~2.6 from~\cite{BM12b} and 
Lemma~\ref{lemma-repetition}.
In either subcase, the existence of $R_j$'s that express the strongest 
post-conditions needed is guaranteed by the expressiveness property of 
$\Lang{\ServAlg}$.
Case~(2) follows directly from the definitions involved.

In case~(1), $\Th{\ServAlg} \Ent \assinseq{b}{P}{S_1 \conc S_2}{e}{Q}$
follows directly by the in\-duc\-tion hypothesis, the first three 
concatenation rules (R1--R3), and the alternatives rule~(R6).
In case~(2), $\Th{\ServAlg} \Ent \assinseq{b}{P}{S_1 \conc S_2}{e}{Q}$
follows directly by the induction hypothesis and the last concatenation 
rule (R4).
\qed
\end{proof}
The next lemma tells us that the axioms and inference rules of the 
presented Hoare-like logic is complete if provability can be identified 
with provability from a particular set of asserted single-pass 
instruction sequences; and 
the second next lemma expresses that the asserted single-pass 
instruction sequences concerned are provable.
\begin{lemma}
\label{lemma-most-general-assinseq-1}
Let $\ServAlg$ be a service algebra such that $\Lang{\ServAlg}$ is 
expressive for $\CTerm{\InSeq}$ on $\ServAlg$.
For each $S \in \CTerm{\InSeq}$,  
let $x_1^S,\ldots,x_{n_S}^S \in \Foci$ and 
$y_1^S,\ldots,y_{n_S}^S \in \Foci$ be such that 
$\vars(S) = \set{x_1^S,\ldots,x_{n_S}^S}$ and 
$\vars(S) \inter \set{y_1^S,\ldots,y_{n_S}^S} = \emptyset$.
For each $S \in \CTerm{\InSeq}$ and $b \in \Natpos$,  
let $P'_S$ be $x_1^S = y_1^S \Land \ldots \Land x_{n_S}^S = y_{n_S}^S$, 
and let $Q'_{S,b} \in \Lang{\ServAlg}$ be such that $Q'_{S,b}$ expresses 
the strongest post-condition of $P'_S$ and $S\rep$ for $b$ and $0$ on 
$\ServAlg$.
For each $S \in \CTerm{\InSeq}$ and $b \in \Natpos$, 
let 
$\nm{ub}_{S,b} = 
 \max \set{b' \in \Natpos \where 
           b' = b \Lor
           \fjmp{b'} \mathrel{\mathrm{occurs\,in}} S}$.
Then, for each $S' \in \CTerm{\InSeq}$, $P,Q \in \Lang{\ServAlg}$, and
$b \in \Natpos$, 
$\ServAlg \Sat \assinseq{b}{P}{S'}{0}{Q}$ implies
$\Th{\ServAlg} \union 
 \set{\assinseq{b'}{P'_S}{S\rep}{0}{Q'_{S,b'}} \where
      b' \leq \nm{ub}_{S',b} \Land 
      S\rep \mathrel{\mathrm{is\,a\,subterm\,of}} S'} 
  \Ent \assinseq{b}{P}{S'}{0}{Q}$.
\end{lemma}
\begin{proof}
\sloppy
This is proved by induction on the structure of $S'$.
The cases that $S'$ is a single instruction follow directly from one of 
the axioms (A2, A5, A8, A10, A11) and the consequence rule (R10).
The case that $S'$ is of the form $S_1 \conc S_2$ is proved, using the 
induction hypothesis, in the same way as the case of concatenation in 
the proof of Lemma~\ref{lemma-completeness}.
What is left is the case that $S'$ is of the form $S\rep$.

In the case that $S'$ is of the form $S\rep$, it suffices to show
that, for each $S \in \CTerm{\InSeq}$, $P,Q \in \Lang{\ServAlg}$, and
$b \in \Natpos$, 
$\ServAlg \Sat \assinseq{b}{P}{S\rep}{0}{Q}$ implies
$\Th{\ServAlg} \union 
 \set{\assinseq{b}{P'_S}{S\rep}{0}{Q'_{S,b}}} \Ent
 \assinseq{b}{P}{S\rep}{0}{Q}$.

Let $S \in \CTerm{\InSeq}$, $P,Q \in \Lang{\ServAlg}$, and 
$b \in \Natpos$, and let $z_1,\ldots,z_{n_S} \in \Foci$ be such that 
$(\vars(S) \union \vars(P) \union \vars(Q) \union
 \set{y_1,\ldots,y_{n_S}}) \inter \set{z_1,\ldots,z_{n_S}} = \emptyset$.
Moreover, 
let $P_1$ be $P\subst{z_1}{y_1^S} \ldots \subst{z_{n_S}}{y_{n_S}^S}$, 
let $Q_1$ be $Q\subst{z_1}{y_1^S}\ldots\subst{z_{n_S}}{y_{n_S}^S}$, and
let $P_2$ be 
$P_1\subst{y_1^S}{x_1^S} \ldots \subst{y_{n_S}^S}{x_{n_S}^S}$.
In the rest of this proof, a \emph{state representing term}
is a closed term of sort $\ServFam$ that is a state representing term 
for $P$, $Q$, $S$, and 
$\set{y_1^S,\ldots,y_{n_S}^S} \union \set{z_1,\ldots,z_{n_S}}$ with 
respect to $\ServAlg$.
Assume $\ServAlg \Sat \assinseq{b}{P}{S\rep}{0}{Q}$.

From $\assinseq{b}{P'_S}{S\rep}{0}{Q'_{S,b}}$, it follows that
$\assinseq{b}{P'_S \Land P_2}{S\rep}{0}{Q'_{S,b} \Land P_2}$~$(*)$
by the invariance rule (R7).
We now show that $\ServAlg \Sat (Q'_{S,b} \Land P_2) \Limpl Q_1$.

Let $t'$ be a state representing term.
Assume $\ServAlg \Sat (Q'_{S,b} \Land P_2)[t']$.
By the definition of $Q'_{S,b}$, there exists a state representing term 
$t$ such that $\ServAlg \Sat P'_S[t]$ and 
$\ISM{\ServAlg} \Sat \extr{S\rep}_{b,0} \sfapply t = t'$ and
$\ISM{\ServAlg} \Sat \extr{S}_{b,e'} \sfapply t = \emptysf$
for all $e' \in \Nat$ with $e \neq e'$.
Suppose $\ServAlg \Sat (\Lnot P_2)[t]$.
From this, the just-mentioned properties of $t$, and the soundness of 
the invariance rule, it follows that $\ServAlg \Sat (\Lnot P_2)[t']$. 
This contradicts the assumption that 
$\ServAlg \Sat (Q'_S \Land P_2)[t']$.
Consequently, $\ServAlg \Sat P_2[t]$.
From this, the first of the above-mentioned properties of $t$, and the 
fact that $\ServAlg \Sat (P'_S \Land P_2) \Limpl P_1$, it follows that 
$\ServAlg \Sat P_1[t]$.
From this, 
the assumption that $\ServAlg \Sat \assinseq{b}{P}{S\rep}{0}{Q}$, and 
the soundness of the substitution rule (R9), it follows that 
$\ServAlg \Sat Q_1[t']$.
This proves that $\ServAlg \Sat (Q'_{S,b} \Land P_2) \Limpl Q_1$~$(**)$.

\sloppy
From $(*)$ and $(**)$, it now follows by the consequence rule (R10) that 
$\assinseq{b}{P'_S \Land P_2}{S\rep}{0}{Q_1}$.
From this, it follows by the elimination rule (R8) that
$\assinseq{b}{\Exists{y_1^S,\ldots,y_{n_S}^S}{(P'_S \Land P_2)}}{S\rep}
          {0}{Q_1}$.
From this and the fact that 
$\ServAlg \Sat P_1 \Limpl
 \Exists{y_1^S,\ldots,y_{n_S}^S}{(P'_S \Land P_2)}$, 
it follows that $\assinseq{b}{P_1}{S\rep}{0}{Q_1}$ by the consequence 
rule.
From this, it follows that $\assinseq{b}{P}{S\rep}{0}{Q}$ by the 
substitution rule.
\qed
\end{proof}

\begin{lemma}
\label{lemma-most-general-assinseq-2}
Let $\ServAlg$ and, for each $S \in \CTerm{\InSeq}$ and $b \in \Natpos$, 
$P'_S$, and $Q'_{S,b}$ be as in 
Lemma~\ref{lemma-most-general-assinseq-1}.
Then, for each $S \in \CTerm{\InSeq}$ and $b \in \Natpos$, 
$\Th{\ServAlg} \Ent \assinseq{b}{P'_S}{S\rep}{0}{Q'_{S,b}}$.

\end{lemma}
\begin{proof}
\sloppy
Let $S \in \CTerm{\InSeq}$ and $b \in \Natpos$.
Then, by the definition of $Q'_{S,b}$, 
$\ServAlg \Sat \assinseq{b}{P'_S}{S\rep}{0}{Q'_{S,b}}$.
From this, it follows that 
$\ServAlg \Sat \assinseq{b}{P'_S}{S \conc S\rep}{0}{Q'_{S,b}}$
because $\extr{S\rep}_{b,0} = \extr{S \conc S\rep}_{b,0}$.
From this and Lemma~\ref{lemma-most-general-assinseq-1}, it follows that
$\Th{\ServAlg} \union
 \set{\assinseq{b'}{P'_S}{S\rep}{0}{Q'_{S,b'}} \where
      b' \leq \nm{ub}_{S \conc S\rep,b}} \Ent
 \assinseq{b}{P'_S}{S \conc S\rep}{0}{Q'_{S,b}}$, 
where $\nm{ub}_{S,b}$ is defined as in 
Lemma~\ref{lemma-most-general-assinseq-1}.
Because we have proved this for an arbitrary $b$, it follows by the
repetition rule that 
$\Th{\ServAlg} \Ent \assinseq{b}{P'_S}{S\rep}{0}{Q'_{S,b}}$.
\qed
\end{proof}
The lines of the proofs of Lemmas~\ref{lemma-most-general-assinseq-1} 
and~\ref{lemma-most-general-assinseq-2}, which are mostly concerned with
repetition, are reminiscent of the lines of the proofs of Lemmas~1 and~2 
from~\cite{Apt81a}, which are mostly concerned with calls of 
(parameterless) recursive procedures.

The following theorem is the completeness theorem for the presented 
Hoare-like logic.
The weak form of completeness that can be proved is known as 
completeness in the sense of Cook because this notion of completeness
originates from Cook~\cite{Coo78a}.
\begin{theorem}
\label{theorem-completeness}
For each service algebra $\ServAlg$ such that $\Lang{\ServAlg}$ is 
expressive for $\CTerm{\InSeq}$ on~$\ServAlg$ and each asserted 
instruction sequence $\phi$, $\ServAlg \Sat \phi$ implies 
$\Th{\ServAlg} \Ent \phi$. 
\end{theorem}
\begin{proof}
\sloppy
This result is an immediate consequence of 
Lemmas~\ref{lemma-repetition}--\ref{lemma-most-general-assinseq-2}.
\qed
\end{proof}

\section{Concluding Remarks}
\label{sect-concl}

We have presented a Hoare-like logic for proving the partial correctness 
of a single-pass instruction sequence as considered in program algebra
and have shown that it is sound and complete in the sense of Cook.
We have extended the asserted programs of Hoare logics with two natural 
numbers which represent conditions on how execution enters and exits an 
instruction sequence.
By that we have prevented that pre- and post-conditions can be 
formulated in which aspects of input-output behaviour and flow of 
execution are combined in ways that are unnecessary for proving 
(partial) correctness of instruction sequences.
We believe that by the way in which we have extended the asserted 
programs of Hoare logics, the presented Hoare-like logic remains as 
close to Hoare logics as possible in the case where program segments 
with multiple entry points and multiple exit points have to be dealt 
with.

In contrast with most related work, we have neither taken ad hoc 
restrictions and features of machine- or assembly-level programs into 
account nor abstracted in an ad hoc way from instruction sequences as 
found in low-level programs.
Moreover, unlike some related work, we have stuck to classical 
first-order logic for pre- and post-conditions.
In particular, the separating conjunction and separating implication
connectives from separation logics~\cite{Rey02a} are not used in pre- 
and post-conditions
Because of this, most related work, including the work reported upon 
in~\cite{JBK13a,MG07a,SU07a}, is only loosely related.

Most closely related is the work reported upon in~\cite{TA06a,Wan76a}.
The form of asserted instruction sequences is inspired 
by~\cite{Wan76a}.
However, as explained in Section~\ref{sect-intro}, their interpretation 
differs somewhat.
Moreover, no attention is paid to soundness and completeness issues 
in~\cite{Wan76a}.
An asserted program from~\cite{TA06a} corresponds essentially to a 
set of asserted instruction sequences concerning the same instruction
sequence fragment.
The particular form of these asserted programs has the effect that 
proofs using the program logic from~\cite{TA06a} involve a lot of 
auxiliary label manipulation.

\bibliographystyle{splncs03}
\bibliography{IS}

\end{document}